\documentclass[preprint,12pt]{elsarticle}
\usepackage[T1]{fontenc}
\usepackage{mathptmx}
\usepackage[scaled=0.9]{helvet}
\usepackage{courier}
\usepackage{amsmath,amssymb,amsthm}
\usepackage{booktabs}
\usepackage{array}
\usepackage{xurl}
\usepackage{algorithm}
\usepackage{algpseudocode}
\usepackage{tikz}
\usetikzlibrary{shapes.geometric,arrows,positioning,calc}
\usepackage[hypertexnames=false,colorlinks=true,breaklinks=true,bookmarks=false,urlcolor=blue,citecolor=blue,linkcolor=blue,bookmarksopen=false,draft=false]{hyperref}

\makeatletter\renewcommand*{\top}{%
{\mathpalette\@transpose{}}%
}
\newcommand*{\@transpose}[2]{%
\raisebox{\depth}{$\m@th#1\mathsf{T}$}%
}
\makeatother

\newtheorem{theorem}{Theorem}
\newtheorem{lemma}[theorem]{Lemma}
\newtheorem{proposition}[theorem]{Proposition}
\newtheorem{corollary}[theorem]{Corollary}
\theoremstyle{remark}
\newtheorem{remark}[theorem]{Remark}

\DeclareMathOperator{\ldet}{ldet}
\newcommand{\R}{\mathbb{R}}
\newcommand{\Nn}{N_n}

\makeatletter\newcommand{\eadname}[1]{\global\@eadauthor={#1}}\makeatother

\journal{Discrete Applied Mathematics}

\begin{document}

\begin{frontmatter}

\title{Faster dynamic programming for\\ tridiagonal maximum-entropy sampling}


\author[ufrj]{Marcia Fampa}
\eadname{M. Fampa}\ead{fampa@cos.ufrj.br}

\author[umich]{Jon Lee}
\eadname{J. Lee}\ead{jonxlee@umich.edu}

\address[ufrj]{Federal University of Rio de Janeiro}

\address[umich]{University of Michigan}

\begin{abstract}
The maximum-entropy sampling problem (MESP) seeks, for an order-$n$ covariance
matrix $C$, a principal submatrix of order $s$ with maximum log-determinant.
Mostly for convenience, we assume that $C$ is nonsingular. Al-Thani and Lee (2023) solved MESP in $O(n^5)$ time when $C$ or
$C^{-1}$ is tridiagonal. We show that the inner maximization of their recursion
depends only on a prefix of the index set and that no piece of a solution is
longer than $s$; this gives an $O(ns^2)$-time algorithm that returns the optimal
value for every budget $t\le s$. When $C^{-1}$ is tridiagonal, $C$ is, up to
scaling, the covariance matrix of an Ornstein--Uhlenbeck process observed at
unevenly spaced times, and MESP becomes choosing points on a line under a concave
gap function with the Monge property; this gives an $O(ns)$-time algorithm and,
in the first-order autoregressive case, a closed-form solution. Given only $C$, we
solve MESP in $O(n^2)$ time whenever $C$ or $C^{-1}$ is tridiagonal, up to a
symmetric permutation, and we recognize these cases within the same bound. For
spiders, we make explicit, and sharpen, the dependence on the number of legs,
and, drawing on a hardness result of Ohsaka for stars, we observe that, unless
$\mathrm{P}=\mathrm{NP}$, the exponent of the running time must grow with the
number of legs.
\end{abstract}

\begin{keyword}
maximum-entropy sampling \sep dynamic programming \sep tridiagonal matrix \sep
determinant \sep Monge property \sep autoregressive process \sep spider \sep nonlinear combinatorial optimization
\end{keyword}

\end{frontmatter}

\section{Introduction}

Let $C$ be an order-$n$ symmetric positive-semidefinite matrix, the covariance
matrix of a Gaussian random vector $Y_{\Nn}$\,, where $\Nn:=\{1,\dots,n\}$, and let
$0<s<n$. We write $\ldet:=\log\det$, and we write $C[S,S]$ for the principal
submatrix of $C$ indexed by $S$. The \emph{maximum-entropy sampling problem}
(MESP)~\cite{SW} is
\[
  z(C,s):=\max\{\ldet C[S,S] : |S|=s,\ S\subset \Nn\}. \tag{MESP}
\]
The input to MESP is the matrix $C$ and the integer $s$. We exclude $s=n$ because
it is trivial: $z(C,n)=\ldet C$, which we compute in $O(n)$ time when $C$ or $C^{-1}$
is tridiagonal (by Lemma~\ref{lem:pivots}, or by the preprocessing of
Section~\ref{sec:precision}). We still use the notation $z(C,t)$ for every
$0\le t\le n$, with $z(C,0)=0$. 
We assume that $C$ is nonsingular.
When we exploit the
structure of $C^{-1}$, we never invert $C$ by a general-purpose method or form a
dense $C^{-1}$, and we account for the cost
of every quantity that we derive from $C$. We state running times in the
arithmetic model, in which each addition, subtraction, multiplication, division,
and comparison of real numbers costs $O(1)$ time; we do not use logarithms or
exponentials as primitive operations (see Section~\ref{sec:prelim}).
MESP is NP-hard~\cite{KLQ}; see~\cite{FL2022} for a comprehensive treatment.
\cite{ALTHANI2023120} (see also, \cite{ALTHANI2021127}) established the first
significant polynomially-solvable case of MESP: when $C$ (or $C^{-1}$, or a
symmetric permutation of either) is tridiagonal, they solve MESP by dynamic
programming in $O(n^5)$ time. Their algorithm is also a building block for
spider-structured matrices and for upper bounds that come from tridiagonal
masks (see \cite{ALTHANI2023120}).

\paragraph{Contribution}
We demonstrate that the recursion of \cite{ALTHANI2023120} performs
redundant work in two places. First, its inner maximization ranges over all
$O(n^2)$ earlier pieces, although it depends only on a prefix of $\Nn$ and on a
budget; therefore, it collapses to a single running maximum, computed once and reused. Second, no
piece is longer than $s$. The resulting algorithm (Algorithm~\ref{alg:new}) runs
in $O(ns^2)$ time and uses $O(ns)$ memory (Theorem~\ref{thm:main}).

In Section~\ref{sec:precision}, we treat the case in which $C^{-1}$ is
tridiagonal directly, rather than by complementation. Every such $C$ is, up to
scaling, the covariance matrix of an Ornstein--Uhlenbeck process observed at
unevenly spaced times --- a classical fact, in several equivalent forms (see
\cite{Barrett1979,Meurant1992,doob1942brownian}) --- and, from this viewpoint, MESP becomes the problem of choosing $s$ of those
times to maximize a sum of node weights and of a concave function of consecutive
gaps. This viewpoint gives a simple closed-form solution in the first-order
autoregressive case (Proposition~\ref{thm:ar1}), an $O(ns)$-time algorithm based on the
Monge property (Theorem~\ref{thm:precision}), and, together with
Theorem~\ref{thm:main}, an $O(n^2)$-time algorithm whenever $C$ or $C^{-1}$ is
tridiagonal (Corollary~\ref{cor:both}).

Two structures ($C$ tridiagonal or $C^{-1}$ tridiagonal) and two families of
methods (that of \cite{ALTHANI2023120} and ours) give four combinations.
Table~\ref{tab:preview} previews the running times and memory requirements; in
each structure, our best bound is $O(n^2)$, compared with $O(n^5)$.

\begin{table}[!htbp]
\centering
\small
\caption{Running time and memory, given $C$ as input. The rows are grouped by
the structure of $C$. Within each group, the column ``Applied to'' gives the matrix
on which the method is run, and ``Budget'' gives the corresponding budget; for
example, when $C$ is tridiagonal, Theorem~\ref{thm:precision} is run on $C^{-1}$,
whose inverse $C$ is tridiagonal, with budget $n-s$. Every bound includes the
$O(n)$ preprocessing that computes the needed entries of $C$ or $C^{-1}$
(Section~\ref{sec:precision}); if it is not known in advance whether $C^{-1}$ is
tridiagonal, deciding this costs $O(n^2)$ time (Proposition~\ref{prop:recognize}, and
Remark~\ref{rem:perm} for symmetric permutations). Both bounds of Corollary~\ref{cor:both} are $O(n^2)$, and both bounds of Al-Thani
and Lee are $O(n^5)$.}\label{tab:preview}
\medskip

\begin{tabular}{@{}lllll@{}}
\toprule
Method & Applied to & Budget & Time & Memory \\
\midrule
\multicolumn{5}{@{}l}{\emph{Structure: $C$ tridiagonal}} \\
Al-Thani and Lee~\cite{ALTHANI2023120} & $C$ & $s$ & $O(n^4s)$ & $O(n^2s)$ \\
Theorem~\ref{thm:main} & $C$ & $s$ & $O(ns^2)$ & $O(ns)$ \\
Theorem~\ref{thm:precision} & $C^{-1}$ & $n-s$ & $O(n(n-s))$ & $O(n(n-s))$ \\
Corollary~\ref{cor:both} & \multicolumn{2}{l}{better of the two} & $O(n\min\{s^2,n-s\})$ & \\
\midrule
\multicolumn{5}{@{}l}{\emph{Structure: $C^{-1}$ tridiagonal}} \\
Al-Thani and Lee~\cite{ALTHANI2023120} & $C^{-1}$ & $n-s$ & $O(n^4(n-s))$ & $O(n^2(n-s))$ \\
Theorem~\ref{thm:main} & $C^{-1}$ & $n-s$ & $O(n(n-s)^2)$ & $O(n(n-s))$ \\
Theorem~\ref{thm:precision} & $C$ & $s$ & $O(ns)$ & $O(ns)$ \\
Corollary~\ref{cor:both} & \multicolumn{2}{l}{better of the two} & $O(n\min\{s,(n-s)^2\})$ & \\
\bottomrule
\end{tabular}
\end{table}

In Section~\ref{sec:spiders}, we turn to spiders. We show that the legs interact
only through the body and, when the body is chosen, only through a single scalar
Schur complement. This observation yields an explicit running time that is linear
in $n$ and exponential only in the number $r$ of legs (Theorem~\ref{thm:spider}),
and it isolates the source of hardness: stars (spiders with $n-1$ legs) are
exactly the case in which that scalar coupling is unrestricted, and MESP is
NP-hard on stars~\cite{Ohsaka}. When the support graph $G(C^{-1})$ of $C^{-1}$ is a
spider instead, the same holds
with budget $n-s$, and we can recognize this case, and find the spider, in $O(n^2)$
time (\ref{app:spider}).

\section{Preliminaries}\label{sec:prelim}

We write $c_i:=C_{ii}$\,. The \emph{support graph} $G(A)$ of a symmetric matrix $A$
of order $n$ has vertex set $\Nn$ and an edge $\{i,j\}$ for each $i\ne j$ with
$A_{ij}\ne0$. In Sections~\ref{sec:prelim}--\ref{sec:faster}, $C$ is tridiagonal,
and we write $b_i:=C_{i,i+1}$\,. For integers $a$ and $i$, we write
$[a,i]:=\{a,\dots,i\}$, which is empty if $i<a$, and $L(a,i):=\ldet C[[a,i],[a,i]]$,
where we take the $\ldet$ of an empty matrix to be $0$.

Following \cite{ALTHANI2023120}, we decompose each $S\subset\Nn$
uniquely into maximal runs of consecutive integers, its \emph{pieces}
$S_1,\dots,S_m$\,, ordered from left to right. At least one index outside $S$
separates consecutive pieces. Because $C$ is tridiagonal, $C[S,S]$ is therefore
block diagonal, and we have the following result.

\begin{lemma}[\cite{ALTHANI2023120}]\label{lem:pieces}
$\ldet C[S,S]=\sum_{j=1}^{m}\ldet C[S_j,S_j]$.
\end{lemma}

We obtain interval log-determinants from the three-term determinant
recursion~\cite{GVL1996}, written as a ratio of consecutive leading principal
minors.

\begin{lemma}\label{lem:pivots}
Let $p(a,a):=c_a$\,, and, for $i>a$, let $p(a,i):=c_i-b_{i-1}^2/p(a,i-1)$. Then
$p(a,i)>0$ and $L(a,i)=L(a,i-1)+\log p(a,i)$.
\end{lemma}

\begin{proof}
The quantities $p(a,i)$ are the pivots of Gaussian elimination, that is, the
ratios $\det C[[a,i],[a,i]]/\det C[[a,i-1],[a,i-1]]$ of consecutive leading
principal minors of $C[[a,i],[a,i]]$. These minors are positive because $C$ is
positive definite.
\end{proof}

For $1\le a\le i\le n$ with $i-a+1\le s$, we obtain all of the values $L(a,i)$
that we need from Lemma~\ref{lem:pivots}, in $O(ns)$ total time.

\paragraph{Model of computation}
We write our recursions with logarithms because they turn products of
determinants into sums, but no algorithm in this paper needs to evaluate a
logarithm or an exponential. Each one only adds and compares sums of logarithms
of positive numbers, and, because $\log$ is increasing, it can instead multiply and
compare the corresponding products, with $0$ in place of $-\infty$ and with the
same number of operations; in particular, every comparison made by SMAWK in
Section~\ref{sec:precision} has the same outcome. Concretely, in
Section~\ref{sec:faster}, the dynamic program becomes a max-times recursion over
the determinants $D(a,i):=\det C[[a,i],[a,i]]$, which the recurrence
$D(a,i)=c_iD(a,i-1)-b_{i-1}^2D(a,i-2)$, with $D(a,a-1):=1$ and $D(a,a-2):=0$, gives
without division. In Section~\ref{sec:precision}, $a_j$ and
$\varphi(\tau_j-\tau_i)$ become $\sigma_j^2$ and $1-x_{ij}$\,, where
$x_{ij}:=e^{-(\tau_j-\tau_i)}=\prod_{k=i}^{j-1}\rho_k^2$ by Lemma~\ref{lem:ou}; from
the prefix products of the $\rho_k^2$, each $x_{ij}$ costs $O(1)$ time. In
Section~\ref{sec:spiders} and \ref{app:spider}, $w_i(e)$ is $\beta_i^2$ times a
ratio of two determinants, and $\ldet C$ becomes $\det C$, a product; in
Remark~\ref{rem:perm} and \ref{app:spider}, we compare the values $d(i,j)$ and
$|\rho_{ij}|$ by comparing the squares $\rho_{ij}^2=C_{ij}^2/(C_{ii}C_{jj})$. The
algorithms thus return an optimal $S$ and $\det C[S,S]=e^{z(C,s)}$. In
floating-point practice, logarithms avoid overflow and underflow, which is why we
keep them in the exposition and in Remark~\ref{rem:crossover}.

\section{The dynamic program of Al-Thani and Lee}

For $1\le k\le\ell\le n$ and $\ell-k+1\le t\le s$, Al-Thani and
Lee~\cite{ALTHANI2023120} define
\[
  f(k,\ell,t):=\max\{\ldet C[S,S] : |S|=t,\ S\subset\Nn,\ \text{the last piece of $S$ is }[k,\ell]\},
\]
with boundary values $f(k,\ell,\ell-k+1)=L(k,\ell)$ and, for $t>\ell-k+1$,
\begin{equation}\label{eq:old}
  \begin{split}
  f(k,\ell,t)=L(k,\ell)+\max_{i,j}\bigl\{&f(i,j,\,t-(\ell-k+1)) :\\
  &1\le i\le j\le k-2,\ j-i+1\le t-(\ell-k+1)\bigr\},
  \end{split}
\end{equation}
where the constraint $j\le k-2$ ensures that index $k-1$ is not chosen, so that
$[k,\ell]$ is a maximal piece. Finally, $z(C,s)=\max_{k,\ell}f(k,\ell,s)$.

\paragraph{Cost}
There are $O(n^2s)$ triples $(k,\ell,t)$, and evaluating \eqref{eq:old} requires
scanning $O(n^2)$ pairs $(i,j)$ for each of them, for a total of $O(n^4s)=O(n^5)$
time, as Al-Thani and Lee~\cite{ALTHANI2023120} state. Even if we restrict
attention to outer pieces with $\ell-k+1\le s$ (so that only $O(ns^2)$ triples
remain) but still scan all $O(n^2)$ pairs $(i,j)$, the cost is $\Theta(n^3s^2)$;
if we also restrict the inner pieces to $j-i+1\le s$ (so that only $O(ns)$ pairs
remain for each triple), the cost is still $\Theta(n^2s^3)$. In either case,
storing $f$ requires $\Theta(ns^2)$ memory.

\section{The faster dynamic program}\label{sec:faster}

\paragraph{Key observation}
The maximum in \eqref{eq:old} depends on $(k,\ell,t)$ only through the pair
$(k-1,t')$, where $t':=t-(\ell-k+1)$: it is the best value of a $t'$-subset of
$[1,k-2]$, that is, of a $t'$-subset of $[1,k-1]$ that omits $k-1$. It does not
depend on $\ell$, and as $k$ increases, the family of admissible pieces $(i,j)$
only grows. Hence, the maximum is a running (prefix) maximum: we can tabulate it
once per pair $(k,t')$, each entry in $O(1)$ time from earlier entries, instead of
recomputing it once per triple $(k,\ell,t)$ in $O(n^2)$ time. The outer
maximization still costs $O(s)$ time per state, so the total is $O(ns^2)$
(Theorem~\ref{thm:main}).

To make this precise, for $0\le i\le n$ and $0\le t\le s$, we define
\begin{align*}
  g(i,t) &:= \max\{\ldet C[A,A] : A\subset[1,i],\ |A|=t,\ i\notin A\},\\
  H(i,t) &:= \max\{\ldet C[A,A] : A\subset[1,i],\ |A|=t,\ i\in A\},
\end{align*}
where the condition $i\notin A$ is vacuous for $i=0$, and where
$\max\emptyset:=-\infty$. Thus, $g$ tabulates the inner maximum, and
$H(\ell,t)=\max_k f(k,\ell,t)$.

\begin{proposition}\label{prop:link}
For all $1\le k\le\ell\le n$ and $\ell-k+1\le t\le s$, we have
$f(k,\ell,t)=L(k,\ell)+g\bigl(k-1,\,t-(\ell-k+1)\bigr)$.
\end{proposition}

\begin{proof}
Let $t':=t-(\ell-k+1)$. If $t'=0$, then $g(k-1,0)=0$ (take $A=\emptyset$), in
agreement with the boundary value. If $t'>0$, then the set $S\setminus[k,\ell]$
ranges exactly over the $t'$-subsets of $[1,k-2]$, whose best value is
$g(k-1,t')$ by definition. By Lemma~\ref{lem:pieces}, $\ldet C[S,S]$ splits
accordingly.
\end{proof}

\begin{theorem}\label{thm:main}
Let $g(0,0):=0$, $g(0,t):=-\infty$ for $t>0$, and $H(0,t):=-\infty$ for $t\ge0$.
Then, for $1\le i\le n$ and $0\le t\le s$,
\begin{align}
  g(i,t) &= \max\{\,g(i-1,t),\ H(i-1,t)\,\},\label{eq:g}\\
  H(i,t) &= \max_{1\le \ell\le\min(t,i)}\bigl\{\,g(i-\ell,\ t-\ell)+L(i-\ell+1,\ i)\,\bigr\},\label{eq:H}
\end{align}
and $z(C,s)=\max\{g(n,s),H(n,s)\}$. Moreover, we can compute these tables in
$O(ns^2)$ time and $O(ns)$ memory, and they yield $z(C,t)$ for every $t\le s$.
\end{theorem}

\begin{proof}
\emph{Correctness of \eqref{eq:g}.} A set $A\subset[1,i]$ with $i\notin A$ is a
subset of $[1,i-1]$ that either omits or contains $i-1$.

\emph{Correctness of \eqref{eq:H}.} If $i\in A$, then $A$ has a unique last piece
$[i-\ell+1,i]$ with $1\le\ell\le\min(t,i)$, and $A':=A\setminus[i-\ell+1,i]$ is a
$(t-\ell)$-subset of $[1,i-\ell]$ that omits $i-\ell$. By
Lemma~\ref{lem:pieces}, $\ldet C[A,A]=\ldet C[A',A']+L(i-\ell+1,i)$, and
maximizing over $A'$ gives $g(i-\ell,t-\ell)$. Conversely, every term in
\eqref{eq:H} is the value of a feasible $A$. (Equivalently, \eqref{eq:H} is the
maximum over $k$ of the identity in Proposition~\ref{prop:link}.) Every
$s$-subset of $\Nn$ either omits or contains $n$; this gives the formula for
$z(C,s)$, and the same argument for $t<s$ gives $z(C,t)$.

\emph{Complexity.} By Lemma~\ref{lem:pivots}, the $O(ns)$ values $L(a,i)$ with
$i-a+1\le s$ cost $O(ns)$ time. Each of the $O(ns)$ entries $g(i,t)$ costs $O(1)$
time, and each of the $O(ns)$ entries $H(i,t)$ requires scanning at most $t\le s$
lengths, for $O(ns^2)$ time in total. We store only $L$, $g$, $H$, and two $O(ns)$
tables of argmax pointers (the maximizing $\ell$ in \eqref{eq:H} and the
maximizing branch in \eqref{eq:g}); from these pointers, we recover an optimal
$S$ by traceback in $O(n)$ time. If only the optimal values are needed, the
working memory beyond the $O(n)$ input drops to $O(s^2)$: because $H(i,\cdot)$
reads $g(i-\ell,\cdot)$ only for $\ell\le s$, we keep just the most recent $s$
rows of $g$ and the most recent row of $H$, and we generate the values
$L(i-\ell+1,i)$ for $\ell=1,\dots,\min(s,i)$ afresh for each $i$, in $O(s)$ time,
by running the pivot recursion of Lemma~\ref{lem:pivots} backward from $i$ (that
is, on $C$ with its indices reversed).
\end{proof}

\begin{algorithm}[!htbp]
\caption{$O(ns^2)$ dynamic program for tridiagonal MESP}\label{alg:new}
\begin{algorithmic}[1]
\Require $c\in\R^n$ and $b\in\R^{n-1}$, defining a tridiagonal
  positive-definite matrix $C$; an integer $0<s<n$
\For{$a=1,\dots,n$} \Comment{interval log-determinants (Lemma~\ref{lem:pivots})}
  \For{$i=a,\dots,\min(n,a+s-1)$}
    \State $p\gets c_a$ if $i=a$, else $p\gets c_i-b_{i-1}^2/p$
    \State $L(a,i)\gets L(a,i-1)+\log p$, where $L(a,a-1):=0$
  \EndFor
\EndFor
\State $g(0,0)\gets 0$;\ $g(0,t)\gets-\infty$ for $t>0$;\ $H(0,t)\gets-\infty$ for $t\ge0$
\For{$i=1,\dots,n$}
  \For{$t=0,\dots,s$}
    \State $g(i,t)\gets\max\{g(i-1,t),H(i-1,t)\}$
    \State $H(i,t)\gets\max_{1\le\ell\le\min(t,i)}\{g(i-\ell,t-\ell)+L(i-\ell+1,i)\}$
  \EndFor
\EndFor
\State \Return $z(C,s)=\max\{g(n,s),H(n,s)\}$ and, by traceback, an optimal $S$
\end{algorithmic}
\end{algorithm}

\section{Tridiagonal precision matrices}\label{sec:precision}

Throughout this section, $M$ is an order-$n$ positive-definite matrix whose inverse
$Q:=M^{-1}$ is tridiagonal. The input remains $C$. When $C^{-1}$ is tridiagonal, we
apply the results of this section with $M:=C$; when $C$ is tridiagonal, we
apply them with $M:=C^{-1}$ and budget $n-s$, through the identity
$z(C,s)=\ldet C+z(C^{-1},n-s)$. In both cases, the algorithms use only $O(n)$
entries of $M$, and we compute these entries from $C$ in $O(n)$ time without
forming $C^{-1}$ (see the paragraph before Corollary~\ref{cor:both}). If $Q_{k,k+1}=0$ for some $k$, then $M$ is
block diagonal, the inverse of each block is tridiagonal, and we treat the blocks
separately and combine them by the knapsack-type recursion of
Remark~\ref{rem:budgets}. We therefore assume that $Q_{k,k+1}\ne0$ for
$1\le k<n$.

In this section, $Y$ denotes a Gaussian random vector with covariance matrix $M$
(so it is the vector of the introduction when $M=C$). We write
$\sigma_i^2:=M_{ii}$ and $\rho_k:=M_{k,k+1}/(\sigma_k\sigma_{k+1})$ for the correlation
of consecutive variables.

The following representation is classical, in several equivalent forms (see
\cite{Barrett1979,Meurant1992,doob1942brownian}); we state it in the form that we
use, and we include a short proof for the convenience of the reader.

\begin{lemma}[Ornstein--Uhlenbeck representation]\label{lem:ou}
Suppose, as above, that $Q_{k,k+1}\ne0$ for $1\le k<n$. Then, for $1\le k<n$, we
have $0<|\rho_k|<1$, and, for $i<j$, the correlation of $Y_i$ and
$Y_j$ is $\prod_{k=i}^{j-1}\rho_k$\,. Consequently, if we let $\tau_1:=0$ and
$\tau_{k+1}:=\tau_k-\log \rho_k^2$ for $1\le k<n$, then
$|\mathrm{corr}(Y_i,Y_j)|=\exp\bigl(-(\tau_j-\tau_i)/2\bigr)$ for all $i<j$.
\end{lemma}

\begin{proof}
For $i<k<j$, the vertex $k$ separates $i$ from $j$ in $G(Q)$, a path, so, by the
global Markov property of Gaussian graphical models~\cite[Corollary~6.4, Proposition~6.6, and Theorem~4.6]{Lauritzen2026}, $Y_i$ and
$Y_j$ are conditionally independent given $Y_k$\,. For Gaussian
variables, this conditional independence means that
$M_{ij}-M_{ik}M_{kj}/M_{kk}$ vanishes, that is,
$\mathrm{corr}(Y_i,Y_j)=\mathrm{corr}(Y_i,Y_k)\,\mathrm{corr}(Y_k,Y_j)$, and induction on
$j-i$ gives the product formula. Positive definiteness of $M[\{k,k+1\},\{k,k+1\}]$
gives $|\rho_k|<1$. If $\rho_k=0$, then the product formula gives $M_{ij}=0$ for all
$i\le k<j$; then $M$, and hence $Q$, is block diagonal, contrary to
$Q_{k,k+1}\ne0$.
\end{proof}

Therefore, up to the signs and scales of its components, the vector $Y$ is a
stationary Ornstein--Uhlenbeck process (see \cite{doob1942brownian}), observed at the times $\tau_1<\dots<\tau_n$
(in general, unevenly spaced). For $x>0$, let $\varphi(x):=\log(1-e^{-x})$. Because
$\varphi'(x)=1/(e^x-1)$, the function $\varphi$ is strictly increasing and strictly
concave on $(0,\infty)$.

The next result follows from Lemma~\ref{lem:ou} and the chain rule for Gaussian
entropy (see \cite{CoverThomas}); cf.\ \cite{Meurant1992} for determinants of
submatrices of matrices with tridiagonal inverse.

\begin{lemma}\label{lem:gaps}
For $S=\{i_1<\dots<i_t\}\subset\Nn$\,, we have
\[
  \ldet M[S,S]=\sum_{i\in S}\log\sigma_i^2+\sum_{k=1}^{t-1}\varphi\bigl(\tau_{i_{k+1}}-\tau_{i_k}\bigr).
\]
\end{lemma}

\begin{proof}
By the chain rule for Gaussian entropy,
$\ldet M[S,S]=\sum_{k=1}^{t}\log\mathrm{Var}(Y_{i_k}\mid Y_{i_1},\dots,Y_{i_{k-1}})$. By
the Markov property, conditioning on $Y_{i_1},\dots,Y_{i_{k-1}}$ is equivalent to
conditioning on $Y_{i_{k-1}}$ alone, and
$\mathrm{Var}(Y_{i_k}\mid Y_{i_{k-1}})=\sigma_{i_k}^2\bigl(1-\mathrm{corr}(Y_{i_{k-1}},Y_{i_k})^2\bigr)
=\sigma_{i_k}^2\bigl(1-e^{-(\tau_{i_k}-\tau_{i_{k-1}})}\bigr)$ by Lemma~\ref{lem:ou}.
\end{proof}

From Lemma~\ref{lem:gaps}, MESP for $M$ is the problem of choosing $s$ of the
points $\tau_1,\dots,\tau_n$ to maximize a sum of node weights $\log\sigma_i^2$ and
of a concave, increasing function of the consecutive gaps. When the node weights
are equal and the points are evenly spaced, an elementary balancing argument gives
a closed form.

\begin{proposition}\label{thm:ar1}
Let $-1<\rho<1$, let $\kappa>0$, and let $M_{ij}:=\kappa\rho^{|i-j|}$, the
covariance matrix of a stationary first-order autoregressive process; its inverse
is $\bigl(\kappa(1-\rho^2)\bigr)^{-1}$ times the tridiagonal matrix with diagonal
$(1,1+\rho^2,\dots,1+\rho^2,1)$ and off-diagonal entries $-\rho$. Then
$z(M,1)=\log\kappa$ and, for $2\le s\le n$, with $q:=\lfloor(n-1)/(s-1)\rfloor$ and
$\theta:=(n-1)-(s-1)q$,
\[
  z(M,s)=s\log\kappa+\theta\log\bigl(1-\rho^{2(q+1)}\bigr)+(s-1-\theta)\log\bigl(1-\rho^{2q}\bigr).
\]
Moreover, if $\rho\ne0$ and $s\ge2$, then $S$ is optimal if and only if $1\in S$,
$n\in S$, and every gap between consecutive elements of $S$ equals $q$ or $q+1$.
\end{proposition}

\begin{proof}
If $\rho=0$, then $M=\kappa I$, and the formula gives $s\log\kappa$ because $q\ge1$.
Let $\rho\ne0$. Then $\sigma_i^2=\kappa$ and $\rho_k=\rho$ for all $i$ and $k$, so
$\tau_k=(k-1)\lambda$ with $\lambda:=-\log\rho^2>0$. By Lemma~\ref{lem:gaps},
$\ldet M[S,S]=s\log\kappa+\sum_{k=1}^{s-1}\psi(d_k)$, where $d_k:=i_{k+1}-i_k$ and
$\psi(d):=\varphi(\lambda d)$ is strictly increasing and strictly concave. The gaps
are positive integers with $\sum_kd_k=i_s-i_1\le n-1$, and every such sequence of
gaps arises from some $S$. Because $\psi$ is strictly increasing, an optimal $S$
has $i_1=1$ and $i_s=n$, so that $\sum_kd_k=n-1$. If two gaps satisfy
$d_k\ge d_{k'}+2$, then replacing them by $d_k-1$ and $d_{k'}+1$ strictly increases
the objective, because strict concavity gives
$\psi(d_{k'}+1)-\psi(d_{k'})>\psi(d_k)-\psi(d_k-1)$. Hence, at an optimum, all gaps lie in
$\{q,q+1\}$, and summing to $n-1$ forces exactly $\theta$ gaps of $q+1$ and $s-1-\theta$ gaps
of $q$; every ordering of these gaps attains the same value.
\end{proof}

The argument above is elementary, and the formula may well be known; we have not
found it in the literature. Because only $\rho^2$ appears, the case $-1<\rho<0$ is
identical to the case $|\rho|$.
When $(s-1)$ divides $(n-1)$, the optimal design is the evenly spaced one, and
$z(M,s)=s\log\kappa+(s-1)\log\bigl(1-\rho^{2(n-1)/(s-1)}\bigr)$.

For general tridiagonal $Q$, the node weights and the spacings vary, and we use
dynamic programming. We write $a_i:=\log\sigma_i^2$ and
$w(i,j):=\varphi(\tau_j-\tau_i)$ for $i<j$.

\begin{theorem}\label{thm:precision}
For $1\le t\le s$ and $t\le j\le n$, let $F(j,t)$ be the greatest value of
$\ldet M[S,S]$ over sets $S$ with $|S|=t$ and $\max S=j$. Then $F(j,1)=a_j$ and,
for $t\ge2$,
\begin{equation}\label{eq:prec}
  F(j,t)=a_j+\max_{t-1\le i<j}\bigl\{F(i,t-1)+w(i,j)\bigr\},
\end{equation}
and $z(M,t)=\max_jF(j,t)$. Moreover, for $i<i'<j<j'$, the gap function satisfies
the Monge inequality  (see \cite{burkard1996perspectives})
\begin{equation}\label{eq:monge}
  w(i,j)+w(i',j')\ge w(i,j')+w(i',j),
\end{equation}
so we can compute $z(M,t)$ for every $t\le s$, and an optimal $S$, in $O(ns)$ time.
\end{theorem}

\begin{proof}
The recursion follows from Lemma~\ref{lem:gaps}, because the last term of the gap
sum involves only $\max S$ and the preceding element of $S$. For
\eqref{eq:monge}, let $\alpha:=\tau_i<\beta:=\tau_{i'}<\gamma:=\tau_j<\delta:=\tau_{j'}$\,.
The pairs $(\gamma-\alpha,\delta-\beta)$ and $(\delta-\alpha,\gamma-\beta)$ have the same
sum, and the second pair is more spread out, because
$\delta-\alpha\ge\max\{\gamma-\alpha,\delta-\beta\}$; concavity of $\varphi$ therefore
gives $\varphi(\gamma-\alpha)+\varphi(\delta-\beta)\ge\varphi(\delta-\alpha)+\varphi(\gamma-\beta)$.

For fixed $t\ge2$, consider the matrix $A$ with rows $j\in[t,n]$ and columns
$i\in[t-1,n-1]$, where $A(j,i):=F(i,t-1)+w(i,j)$ if $i<j$ (a \emph{defined} entry)
and $A(j,i):=-K(i-j+1)^2$ if $i\ge j$, for a constant $K$ that we choose below. We
claim that
\begin{equation}\label{eq:mongeA}
  A(j,i)+A(j',i')\ge A(j,i')+A(j',i)\quad\text{for all } j<j' \text{ and } i<i'.
\end{equation}
Let $B$ be the greatest absolute value of a defined entry, and let $K\ge3B+1$.
Because $(j',i)$ is defined whenever any of the other three entries is, and
$(j',i')$ is defined whenever $(j,i')$ is, exactly six patterns of defined entries
can occur. We write $h(\mu):=-K(\mu+1)^2$, so that $h(\mu)\le-K$ for $\mu\ge0$, and we use
the fact that $h(\mu)-h(\mu+\nu)=K\bigl((\mu+\nu+1)^2-(\mu+1)^2\bigr)\ge2K$ for
$\mu\ge0$ and $\nu\ge1$.
\begin{enumerate}
\item[(1)] All four entries are defined ($i'<j$). The terms $F(i,t-1)$ and
$F(i',t-1)$ cancel, and \eqref{eq:mongeA} is \eqref{eq:monge}.
\item[(2)] Only $(j,i')$ is undefined. The left-hand side is at least $-2B$, and the
right-hand side is at most $-K+B$, which is less than $-2B$.
\item[(3)] Exactly $(j,i')$ and $(j',i')$ are undefined. Then \eqref{eq:mongeA}
reads $A(j,i)+h(i'-j')\ge h(i'-j)+A(j',i)$, and, because $i'-j=(i'-j')+(j'-j)$, the
difference of the two sides is at least $2K-2B>0$.
\item[(4)] Exactly $(j,i)$ and $(j,i')$ are undefined. Symmetrically, the inequality
reads $h(i-j)+A(j',i')\ge h(i'-j)+A(j',i)$, and, because $i'-j=(i-j)+(i'-i)$, the
difference is at least $2K-2B>0$.
\item[(5)] Only $(j',i)$ is defined. With $\mu:=i-j\ge0$, $\nu:=i'-j'\ge0$, and
$\xi:=j'-i\ge1$, we have $i'-j=\mu+\nu+\xi$, and the difference of the two sides is
$K\bigl((\mu+\nu+\xi+1)^2-(\mu+1)^2-(\nu+1)^2\bigr)-A(j',i)\ge2K(\mu+1)(\nu+1)-B>0$.
\item[(6)] No entry is defined. With $\mu:=i-j'\ge0$, $\nu:=j'-j\ge1$, and
$\xi:=i'-i\ge1$, the difference of the two sides is
$K\bigl((\mu+\nu+\xi+1)^2+(\mu+1)^2-(\mu+\nu+1)^2-(\mu+\xi+1)^2\bigr)=2K\nu\xi>0$.
\end{enumerate}
Inequality \eqref{eq:mongeA} implies that, for $j<j'$, if column $i'$ is at least as
good as column $i<i'$ in row $j$, then it is also at least as good in row $j'$;
that is, $A$ is totally monotone for row maxima. Every row contains a defined
entry (column $t-1<j$), and every defined entry exceeds every undefined one, so the
row maxima of $A$ are the maxima in \eqref{eq:prec}. The so-called ``SMAWK
algorithm'' (see \cite{SMAWK}) finds all row maxima of $A$ with $O(n)$ evaluations of
entries. In an implementation, we compare entries lexicographically, by the pair
$\bigl(-(i-j+1)^2,0\bigr)$ for an undefined entry and $(0,A(j,i))$ for a defined one,
so $K$ and $B$ never need to be computed. This is valid because SMAWK accesses $A$
only through comparisons of two entries in the same row, and, for every
$K\ge3B+1$, each such comparison has the same outcome under both orders: two
defined entries compare by value; a defined entry exceeds an undefined one, because
$-B>-K\ge-K(i-j+1)^2$; and two undefined entries compare by $-(i-j+1)^2$ in both.
Hence, SMAWK makes the same choices as on $A$, for which the argument above holds
for every $K\ge3B+1$. Over $t=2,\dots,s$, the total time is
$O(ns)$, and storing the maximizing $i$ for each $(j,t)$ gives an optimal $S$ by
traceback.
\end{proof}

\paragraph{Input and preprocessing}
The quantities $a_i$ and $\tau_i$ depend only on the diagonal and the first
superdiagonal of $M$, and we obtain all of the entries that we need from the input
$C$ in $O(n)$ time.
\begin{enumerate}
\item[(a)] If $C^{-1}$ is tridiagonal and $M:=C$, then we read the $2n-1$ entries
$C_{kk}$ and $C_{k,k+1}$ and invert nothing. If, instead, we wish to apply
Theorem~\ref{thm:main} to $C^{-1}$, then Proposition~\ref{prop:recognize} gives
all of $C^{-1}$, namely its $2n-1$ nonzero entries, and $\ldet C$ from the same
entries of $C$ in $O(n)$ time.
\item[(b)] If $C$ is tridiagonal and $M:=C^{-1}$, then, with forward pivots
$p_1:=C_{11}$\,, $p_i:=C_{ii}-C_{i-1,i}^2/p_{i-1}$ (that is, $p_i=p(1,i)$ in
Lemma~\ref{lem:pivots}) and backward pivots $\bar p_n:=C_{nn}$\,,
$\bar p_i:=C_{ii}-C_{i,i+1}^2/\bar p_{i+1}$\,, we have $M_{ii}=1/(p_i+\bar p_i-C_{ii})$,
$M_{i,i+1}=-C_{i,i+1}M_{i+1,i+1}/p_i$\,, and $\ldet C=\sum_i\log p_i$\,.
\end{enumerate}
Forming a dense $C^{-1}$, or inverting by a general-purpose method, would be
wasteful. When $C$ is tridiagonal, $C^{-1}$ is,
in general, dense (by Lemma~\ref{lem:ou}, every correlation is a product of nonzero
factors), so merely writing it down costs $\Theta(n^2)$ time, at least as much as the
$O(n(n-s))$ cost of the entire dynamic program of Theorem~\ref{thm:precision}; a
general-purpose inversion costs $\Theta(n^3)$ time, more than the $O(ns^2)$ cost
of applying Theorem~\ref{thm:main} to $C$ directly, for every $s$. When $C^{-1}$ is
tridiagonal, $C$ is dense, and a general-purpose inversion of $C$ would again cost
$\Theta(n^3)$ time, far more than the $O(ns)$ cost of Theorem~\ref{thm:precision}.
We need not know in advance whether $C^{-1}$ is tridiagonal: in $O(n^2)$ time,
which is linear in the size of the input $C$, we can determine whether $C$ is
nonsingular with a tridiagonal inverse and, if it is, write down $C^{-1}$. The
following characterization is classical in other forms (see
\cite{Barrett1979} and the review \cite{Meurant1992}); we state
it in the form that we use and include a short proof.

\begin{proposition}\label{prop:recognize}
Let $C$ be an order-$n$ symmetric positive-semidefinite matrix. Then $C$ is
nonsingular and $C^{-1}$ is tridiagonal if and only if the following two conditions
hold:
\begin{enumerate}
\item[(i)] $C_{kk}>0$ for $1\le k\le n$, and $v_k>0$ for $1\le k\le n$, where
$v_1:=C_{11}$ and $v_{k+1}:=C_{k+1,k+1}-C_{k,k+1}^2/C_{kk}$;
\item[(ii)] $C_{ij}\,C_{j-1,j-1}=C_{i,j-1}\,C_{j-1,j}$ for all $1\le i\le j-2\le n-2$.
\end{enumerate}
In that case, with $\beta_k:=C_{k,k+1}/C_{kk}$\,, the nonzero entries of $Q:=C^{-1}$ are
$Q_{kk}=1/v_k+\beta_k^2/v_{k+1}$ (without the second term when $k=n$) and
$Q_{k,k+1}=Q_{k+1,k}=-\beta_k/v_{k+1}$\,, and $\ldet C=\sum_{k=1}^{n}\log v_k$\,.
Consequently, we can check (i) and (ii), and, if they hold, write down $C^{-1}$ and
$\ldet C$, in $O(n^2)$ time.
\end{proposition}

\begin{proof}
Suppose that $C$ is nonsingular, hence positive definite, and that $C^{-1}$ is
tridiagonal. Then $C_{kk}>0$, and $v_{k+1}=\det C[\{k,k+1\},\{k,k+1\}]/C_{kk}>0$, so
(i) holds. Condition~(ii) holds by the conditional-independence argument in the
proof of Lemma~\ref{lem:ou}, because $j-1$ separates $i$ from $j$ in $G(C^{-1})$, a
subgraph of a path.

Conversely, suppose that (i) and (ii) hold. Let $X_1$ be Gaussian with variance
$v_1$\,, and let $X_{k+1}:=\beta_kX_k+\varepsilon_{k+1}$ for $1\le k<n$, where
$\varepsilon_2,\dots,\varepsilon_n$ are independent Gaussian random variables,
independent of $X_1$\,, with variances $v_2,\dots,v_n$\,. Let $\hat C$ be the covariance
matrix of $X$. By induction on $k$, we have
$\hat C_{k+1,k+1}=\beta_k^2\hat C_{kk}+v_{k+1}=C_{k+1,k+1}$ and
$\hat C_{k,k+1}=\beta_k\hat C_{kk}=C_{k,k+1}$\,. For $j\ge i+2$, we have
$\hat C_{ij}=\beta_{j-1}\hat C_{i,j-1}=C_{j-1,j}\hat C_{i,j-1}/C_{j-1,j-1}$\,, and (ii) gives the same
recursion for $C_{ij}$\,; hence, by induction on $j-i$, we have $\hat C=C$. Writing
$\varepsilon_1:=X_1$\,, we have $EX=\varepsilon$, where $E$ is unit lower bidiagonal
with subdiagonal entries $-\beta_k$\,, so
$C^{-1}=\hat C^{-1}=E^{\top}\mathrm{Diag}(v)^{-1}E$, which is positive definite,
tridiagonal, and has the stated entries; moreover, $\det C=\prod_kv_k$ because
$\det E=1$.

Condition~(i) costs $O(n)$ time, condition~(ii) consists of $O(n^2)$ tests that
each cost $O(1)$ time, and the formulas for $C^{-1}$ and $\ldet C$ cost $O(n)$
time.
\end{proof}

\begin{corollary}\label{cor:both}
Given $C$ as input, if $C^{-1}$ is known to be tridiagonal, then we can compute
$z(C,s)$ in $O\bigl(n\min\{s,(n-s)^2\}\bigr)$ time, and if $C$ is tridiagonal,
then we can compute $z(C,s)$ in $O\bigl(n\min\{s^2,n-s\}\bigr)$ time;
these bounds include all preprocessing. If it is not known in advance whether
$C^{-1}$ is tridiagonal, then deciding this costs an additional $O(n^2)$ time
(Proposition~\ref{prop:recognize}). More generally, if $C$ or $C^{-1}$ is a
symmetric permutation of a tridiagonal matrix, then finding such a permutation, or
deciding that none exists, costs $O(n^2)$ time (Remark~\ref{rem:perm}), after which
the same bounds apply. In all cases, the running time is $O(n^2)$.
\end{corollary}

\begin{proof}
If $C^{-1}$ is tridiagonal, then we apply Theorem~\ref{thm:precision} with $M:=C$,
or Theorem~\ref{thm:main} to $C^{-1}$ with budget $n-s$, using~(a). If $C$ is
tridiagonal, then we apply Theorem~\ref{thm:main} to $C$, or
Theorem~\ref{thm:precision} with $M:=C^{-1}$ and budget $n-s$, using~(b). In each
case, the preprocessing costs $O(n)$ time. When $M$ is block diagonal, merging the
tables of the blocks costs, in total, at most the budget times the sum of the block
orders; this cost is within the same bounds. Finally, $\min\{s,(n-s)^2\}\le n$ and
$\min\{s^2,n-s\}\le n$.
\end{proof}

\begin{remark}[Inverse and permuted tridiagonal matrices]\label{rem:perm}
When $C^{-1}$ is tridiagonal, Algorithm~\ref{alg:new}, applied to $C^{-1}$ with
budget $n-s$, runs in $O(n(n-s)^2)$ time, by the identity $z(C,s)=\ldet C+z(C^{-1},n-s)$ (see \cite{ALTHANI2023120}); Corollary~\ref{cor:both}
records how to combine this with Theorem~\ref{thm:precision}.

Symmetric permutations need more care. If $C$ is a symmetric permutation of a
tridiagonal matrix, then we recognize that its support graph is a disjoint union
of paths, and find the order along each path, in time linear in the size of $C$.
If, instead, $C^{-1}$ is a symmetric permutation of a tridiagonal matrix, then the
support graph of $C^{-1}$ is not visible from the dense input $C$, but
Lemma~\ref{lem:ou} recovers it. The blocks of $C^{-1}$ are the connected
components of the support graph of $C$, and, by Lemma~\ref{lem:ou}, every entry of
$C$ within a block is nonzero; so we compute the components in $O(n^2)$ time, and
we reject if some component is not a clique. Within a component, by
Lemma~\ref{lem:ou},
$d(i,j):=-\log\bigl(C_{ij}^2/(C_{ii}C_{jj})\bigr)$ is the distance between the
points $\tau$ of $i$ and $j$ in the unknown order along the path, so $d$ is a line
metric on distinct points. Hence, the vertex farthest from an arbitrary vertex is
an endpoint of the path, and sorting the component by distance from that endpoint,
in $O(n\log n)$ time, recovers the order along the path, up to reversal. Finally,
for each component, we test the conditions of Proposition~\ref{prop:recognize} in
this order, in $O(n^2)$ total time. If they hold, then $C^{-1}$ is tridiagonal in
the resulting order; if they fail, then no symmetric permutation of $C^{-1}$ is
tridiagonal, because any such permutation would induce the same order, up to
reversal, on each component. Thus, recognizing the case in which $C^{-1}$ is a
permuted tridiagonal matrix fits within the same $O(n^2)$ bound as
Proposition~\ref{prop:recognize}.
\end{remark}

Figure~\ref{fig:recipe} collects the resulting recipe: which structure to test,
which regime of $s$ we are in, and which method to apply in each case.

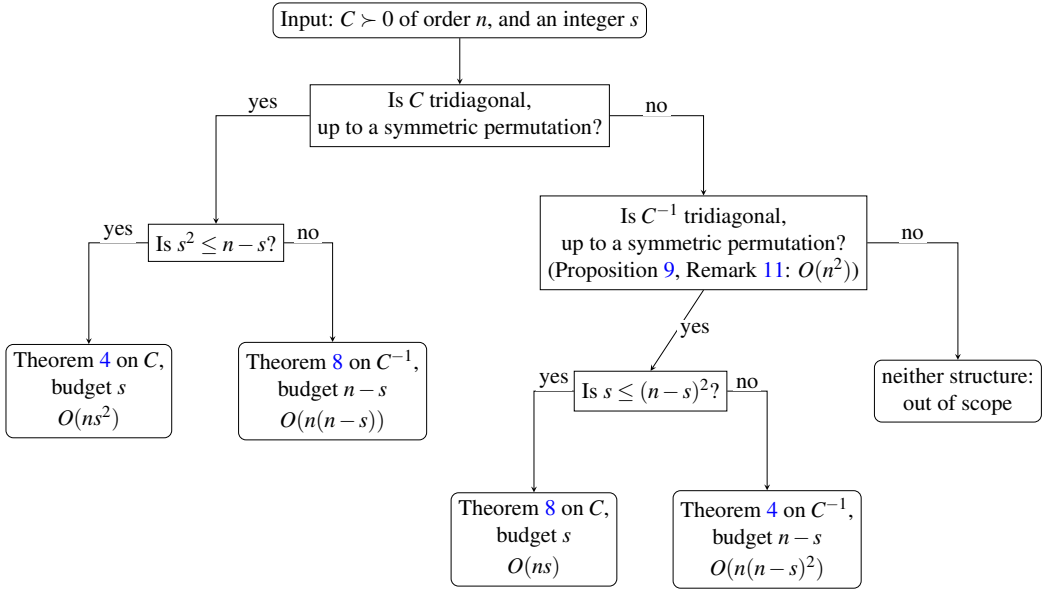
\begin{figure}[!htbp]
\centering
\resizebox{\textwidth}{!}{%
\begin{tikzpicture}[>=stealth,
  box/.style={draw, rounded corners, align=center, inner sep=4pt},
  dec/.style={draw, align=center, inner sep=4pt},
  lbl/.style={inner sep=1.5pt, fill=white}]
\node[box] (in) at (0,1.8) {Input: $C\succ0$ of order $n$, and an integer $s$};
\node[dec] (d0) at (0,0) {Is $C$ tridiagonal,\\ up to a symmetric permutation?};
\node[dec] (dl) at (-4.6,-2.4) {Is $s^2\le n-s$?};
\node[dec] (dr) at (4.6,-2.4) {Is $C^{-1}$ tridiagonal,\\ up to a symmetric permutation?\\ (Proposition~\ref{prop:recognize}, Remark~\ref{rem:perm}: $O(n^2)$)};
\node[box] (l1) at (-7.0,-5.2) {Theorem~\ref{thm:main} on $C$,\\ budget $s$\\[2pt] $O(ns^2)$};
\node[box] (l2) at (-2.4,-5.2) {Theorem~\ref{thm:precision} on $C^{-1}$,\\ budget $n-s$\\[2pt] $O(n(n-s))$};
\node[dec] (dr2) at (3.6,-5.2) {Is $s\le(n-s)^2$?};
\node[box] (out) at (9.4,-5.2) {neither structure:\\ out of scope};
\node[box] (r1) at (1.4,-8.0) {Theorem~\ref{thm:precision} on $C$,\\ budget $s$\\[2pt] $O(ns)$};
\node[box] (r2) at (5.8,-8.0) {Theorem~\ref{thm:main} on $C^{-1}$,\\ budget $n-s$\\[2pt] $O(n(n-s)^2)$};
\draw[->] (in) -- (d0);
\draw[->] (d0.west) -| node[lbl,pos=0.25,above]{yes} (dl.north);
\draw[->] (d0.east) -| node[lbl,pos=0.25,above]{no} (dr.north);
\draw[->] (dl.west) -| node[lbl,pos=0.25,above]{yes} (l1.north);
\draw[->] (dl.east) -| node[lbl,pos=0.25,above]{no} (l2.north);
\draw[->] (dr.east) -| node[lbl,pos=0.25,above]{no} (out.north);
\draw[->] (dr.south) -- node[lbl,right]{yes} (dr2.north);
\draw[->] (dr2.west) -| node[lbl,pos=0.25,above]{yes} (r1.north);
\draw[->] (dr2.east) -| node[lbl,pos=0.25,above]{no} (r2.north);
\end{tikzpicture}}
\caption{How to solve MESP in $O(n^2)$ time when $C$ or $C^{-1}$ is tridiagonal,
up to a symmetric permutation; in the permuted case, we first reorder $C$ by the
permutation found as in Remark~\ref{rem:perm}. Each branch also needs the $O(n)$ preprocessing described above: the band of
$C^{-1}$ and $\ldet C$ for the branch that applies Theorem~\ref{thm:precision} to
$C^{-1}$, and the tridiagonal matrix $C^{-1}$ itself for the branch that applies
Theorem~\ref{thm:main} to $C^{-1}$. If the matrix in use is block diagonal, then we
treat the blocks separately and combine them as in Remark~\ref{rem:budgets}.}
\label{fig:recipe}
\end{figure}

\begin{remark}[Neither method dominates]\label{rem:crossover}
When $C$ is tridiagonal, Theorem~\ref{thm:main} applies to $C$ itself, with budget
$s$, at a cost of $O(ns^2)$. Theorem~\ref{thm:precision} does not apply to $C$
itself, because a tridiagonal $C$ does not, in general, have a tridiagonal inverse;
it applies only to $C^{-1}$, with budget $n-s$, at a cost of $O(n(n-s))$. The first
method is faster when $s^2<n-s$, that is, roughly when $s<\sqrt{n}$ (for example,
$O(n)$ versus $O(n^2)$ time when $s=1$), and the second method is faster otherwise
(for example, $O(n^2)$ versus $O(n^3)$ time when $s=n/2$). Moreover,
Theorem~\ref{thm:main} yields $z(C,t)$ for every $t\le s$, whereas the route through
$C^{-1}$ yields $z(C,t)$ for every $t\ge s$. Symmetrically, when $C^{-1}$ is
tridiagonal, Theorem~\ref{thm:precision} is faster unless $(n-s)^2<s$.

These crossovers are asymptotic, and the constants move them substantially. In
matched implementations of the two dynamic programs, we observed crossovers near
$s\approx8\sqrt{n-s}$ and, symmetrically, near $n-s\approx8\sqrt{s}$; that is, the
range in which Theorem~\ref{thm:main} is preferable is about eight times wider, in
$s$, than the asymptotic rule suggests. Two specific effects account for much of this. First, an
inner step of Theorem~\ref{thm:main} is one addition and one comparison of tabulated
values, whereas an inner step of Theorem~\ref{thm:precision} evaluates
$\varphi(\tau_j-\tau_i)$, and the arguments $\tau_j-\tau_i$ are all distinct, so
nothing can be tabulated; still, in our implementation, a logarithm or an
exponential costs only about twice an addition, so this effect alone is small. Second, and
more important, the SMAWK algorithm evaluates roughly $13$ entries per row of each
layer, rather than the one evaluation per state of Theorem~\ref{thm:main}. The
product of these two factors, about $25$, enters the crossover under a square root.
\end{remark}

\section{Spiders with few legs}\label{sec:spiders}

Following~\cite{ALTHANI2023120}, a \emph{spider} with $r\ge3$ legs
is a tree with one vertex of degree $r$, the \emph{body}, such that deleting the
body leaves $r$ paths, the \emph{legs}. Suppose that $G(C)$ is a spider with body $1$ and legs of orders $n_1,\dots,n_r\ge1$, so that
$n=1+\sum_in_i\ge4$. We index the
vertices of leg~$i$ as $v_{i,1},\dots,v_{i,n_i}$\,, in order of increasing distance
from the body, and we write $\beta_i:=C_{1,v_{i,1}}$ for the entry that joins the
body to leg~$i$. For $0\le e\le n_i$\,, we let $E_i(e):=\{v_{i,1},\dots,v_{i,e}\}$
and $T_i(e):=C[E_i(e),E_i(e)]$, a tridiagonal matrix. As in Section~\ref{sec:prelim},
$L_i(a,e)$ denotes the log-determinant of $C$ on $\{v_{i,a},\dots,v_{i,e}\}$, with
$L_i(a,a-1):=0$.

\begin{lemma}\label{lem:body}
Let $S\subset\Nn$\,.
\begin{enumerate}
\item[(a)] If $1\notin S$, then $\ldet C[S,S]=\sum_{i=1}^{r}\ldet C[S\cap V_i,\,S\cap V_i]$,
where $V_i$ is the vertex set of leg~$i$.
\item[(b)] If $1\in S$, let $e_i$ be the greatest $e$ with $E_i(e)\subset S$, let
$P:=\{1\}\cup\bigcup_i E_i(e_i)$, and let $R_i:=(S\cap V_i)\setminus E_i(e_i)$.
Then $\ldet C[S,S]=\ldet C[P,P]+\sum_{i=1}^{r}\ldet C[R_i,R_i]$ and
\begin{equation}\label{eq:body}
  \ldet C[P,P]=\sum_{i=1}^{r}L_i(1,e_i)+\log\Bigl(c_1-\sum_{i=1}^{r}w_i(e_i)\Bigr),
\end{equation}
where $w_i(0):=0$ and, for $e\ge1$, $w_i(e):=\beta_i^2\,[T_i(e)^{-1}]_{11}
=\beta_i^2\exp\bigl(L_i(2,e)-L_i(1,e)\bigr)$; here, each $T_i(e_i)$ is nonsingular and the
argument of the logarithm is positive, because $C[P,P]$ is positive definite.
\end{enumerate}
\end{lemma}

\begin{proof}
(a) Without the body, no nonzero entry of $C$ joins two legs, so $C[S,S]$ is block
diagonal over the legs.
(b) For each $i$, if $e_i<n_i$\,, then vertex $v_{i,e_i+1}$ is not in $S$, and if
$e_i=n_i$\,, then $R_i$ is empty. In either case, as in Lemma~\ref{lem:pieces}, no
nonzero entry of $C$ joins $R_i$ to $P$ or to $R_{i'}$ for $i'\ne i$, and $C[S,S]$ is block
diagonal with blocks $C[P,P]$ and $C[R_i,R_i]$. We order $P$ with the body first and write $C[P,P]=\bigl(\begin{smallmatrix}c_1&y^{\top}\\
y&\Delta\end{smallmatrix}\bigr)$ with $\Delta:=\mathrm{Diag}\bigl(T_1(e_1),\dots,T_r(e_r)\bigr)$,
where $y$ has entry $\beta_i$ in the first position of block~$i$ (when $e_i\ge1$)
and zeros elsewhere. Taking the Schur complement of $\Delta$ gives
$\det C[P,P]=\det\Delta\,\bigl(c_1-y^{\top}\Delta^{-1}y\bigr)$; because $C[P,P]$ and its
principal submatrix $\Delta$ are positive definite, $c_1-y^{\top}\Delta^{-1}y>0$. Moreover,
$y^{\top}\Delta^{-1}y=\sum_i\beta_i^2[T_i(e_i)^{-1}]_{11}$\,. Finally,
$[T_i(e)^{-1}]_{11}$ is the ratio of the determinant of $T_i(e)$ with its first row
and column deleted to $\det T_i(e)$; this identity gives the exponential form of
$w_i(e)$.
\end{proof}

From Lemma~\ref{lem:body}, the legs interact in only two ways: through the total
budget and, when the body is chosen, through the single scalar
$\sum_i w_i(e_i)$ inside one logarithm. Everything else is separable across legs.

\paragraph{Leg tables}
For each leg~$i$, for $0\le m\le n_i$\,, and for $0\le t\le s$, let $Z_i(m,t)$ be the
greatest value of $\ldet C[A,A]$ over sets $A\subset\{v_{i,m+1},\dots,v_{i,n_i}\}$
with $|A|=t$; for $m\ge n_i$\,, we set $Z_i(m,0):=0$ and $Z_i(m,t):=-\infty$ for
$t>0$. If we list the vertices of leg~$i$ in reverse order, from the tip toward the
body, then the sets $\{v_{i,m+1},\dots,v_{i,n_i}\}$ are exactly the prefixes, so
Theorem~\ref{thm:main}, applied to this reversed path, gives
$Z_i(m,t)=\max\{g(n_i-m,t),H(n_i-m,t)\}$ for all $m$ and $t$ at once, in
$O(n_is^2)$ time. The same pass, together with Lemma~\ref{lem:pivots}, gives every
$L_i(1,e)$, $L_i(2,e)$, and $w_i(e)$ with $e\le s$ in $O(n_is)$ time.

\begin{theorem}\label{thm:spider}
Suppose that $G(C)$ is a spider with $r$ legs, and let
$\mathcal{E}:=\{e\in\mathbb{Z}^r : 0\le e_i\le n_i\,,\ \sum_ie_i\le s-1\}$. Then
$z(C,s)=\max\{z_{\mathrm{out}},z_{\mathrm{in}}\}$, where
\begin{align*}
  z_{\mathrm{out}}&:=\max\Bigl\{\sum_{i=1}^{r}Z_i(0,t_i) : \sum_it_i=s,\ t\ge0\Bigr\},\\
  z_{\mathrm{in}}&:=\max_{e\in\mathcal{E}}\Bigl\{\sum_{i=1}^{r}L_i(1,e_i)
  +\log\Bigl(c_1-\sum_{i=1}^{r}w_i(e_i)\Bigr)\\
  &\qquad\qquad+\max\Bigl\{\sum_{i=1}^{r}Z_i(e_i+1,t_i) : \sum_it_i=s-1-\sum_ie_i\,,\ t\ge0\Bigr\}\Bigr\}.
\end{align*}
We can compute $z(C,s)$ in
$O\bigl(ns^2+s^2\min\bigl\{r|\mathcal{E}|,\binom{s+r-1}{r-1}\bigr\}\bigr)$ time,
where $|\mathcal{E}|\le\binom{s+r-1}{r}$; for fixed $r$, this is
$O(ns^2+s^{r+1})$.
\end{theorem}

\begin{proof}
The two formulas follow from Lemma~\ref{lem:body}: in case~(a), the legs are
independent, and in case~(b), once we fix $e$, the sets $R_i$ are arbitrary subsets
of $\{v_{i,e_i+2},\dots,v_{i,n_i}\}$ whose sizes sum to $s-1-\sum_ie_i$\,. Every inner
maximum over $(t_1,\dots,t_r)$ is a knapsack-type merge of $r$ tables of length
$s+1$, built by $r-1$ pairwise merges; a pairwise merge costs $O(s^2)$ time, or
$O(s)$ time if we need it at only one budget. The leg tables cost
$O(\sum_in_is^2)=O(ns^2)$ time, and $z_{\mathrm{out}}$ costs $O(rs^2)$ time.

For $z_{\mathrm{in}}$\,, we enumerate $\mathcal{E}$ in lexicographic order, by
depth-first search over the prefixes $(e_1,\dots,e_j)$ of its elements, and we keep,
for the current prefix, the merge of $Z_1(e_1+1,\cdot),\dots,Z_j(e_j+1,\cdot)$
together with the partial sums of $L_i(1,e_i)$ and $w_i(e_i)$. Extending a prefix of
length $j-1\le r-2$ by one component costs one pairwise merge, $O(s^2)$ time. For a
complete $e\in\mathcal{E}$, the last merge is needed only at the single budget
$s-1-\sum_ie_i$\,, so it costs $O(s)$ time, and \eqref{eq:body} costs $O(1)$ time
from the partial sums. Prefixes of length $j$ are nonnegative integer vectors with
$j$ components that sum to at most $s-1$, so there are at most $\binom{s-1+j}{j}$
of them, and $\sum_{j=0}^{r-1}\binom{s-1+j}{j}=\binom{s+r-1}{r-1}$; moreover, each
extends (by zeros) to an element of $\mathcal{E}$, so there are at most
$|\mathcal{E}|$ of them for each $j$. For $j=r$, the same count gives
$|\mathcal{E}|\le\binom{s-1+r}{r}=\frac{s}{r}\binom{s+r-1}{r-1}$, so the $O(s)$ cost
per element of $\mathcal{E}$ is within the stated bound.
\end{proof}

\cite{ALTHANI2023120} established polynomial solvability for fixed
$r$ by counting $O(n^r)$ pieces that contain the body. Theorem~\ref{thm:spider}
makes the dependence explicit: the running time is linear in $n$, the dependence
on $r$ is confined to the body-containing case, and the case $1\notin S$ costs only
$O(ns^2)$ time for every $r$. By the identity $z(C,s)=\ldet C+z(C^{-1},n-s)$, the
same conclusions hold, with budget $n-s$, when $G(C^{-1})$ is a known spider.
Theorem~\ref{thm:spider} applied to $C^{-1}$ reads only the diagonal of $C^{-1}$ and
its entries on the edges of the spider, and these, together with $\ldet C$, cost
$O(n)$ time to compute from the entries of $C$ in the same positions, by the
standard formulas for a tree-structured precision matrix~\cite[Eqs.~(6.58)--(6.59)
and Section~8.2.1]{Lauritzen2026}:
\begin{align*}
  C^{-1}&=\sum_{\{u,v\}}\bigl(C[\{u,v\},\{u,v\}]\bigr)^{-1}
  -\mathrm{Diag}\Bigl(\frac{\deg(v)-1}{C_{vv}}\Bigr)_{v\in\Nn}\,,\\
  \ldet C&=\sum_{\{u,v\}}\ldet C[\{u,v\},\{u,v\}]-\sum_{v\in\Nn}(\deg(v)-1)\log C_{vv}\,,
\end{align*}
where the sums run over the edges $\{u,v\}$ of the spider, each $2\times2$ inverse
is padded with zeros to order $n$, and $\deg(v)$ is the degree of $v$ in the
spider. If the spider is not known, then recognizing it, and computing the same
quantities, costs $O(n^2)$ time (\ref{app:spider}).

\paragraph{Stars and hardness}
A \emph{star} is a spider in which every leg has order one, so that $r=n-1$ and $C$
is an arrowhead matrix~\cite{ALTHANI2023120}. Writing $d_i$ for the diagonal entry of
leaf~$i$, we have $w_i(1)=\beta_i^2/d_i$\,, and \eqref{eq:body} reduces to
\[
  \ldet C[\{1\}\cup J,\{1\}\cup J]=\sum_{i\in J}\log d_i
  +\log\Bigl(c_1-\sum_{i\in J}\beta_i^2/d_i\Bigr)
\]
for every set $J$ of leaves. The case $1\notin S$ is trivial on a star (we choose the
$s$ leaves with the greatest $d_i$), so all of the difficulty lies in the scalar
coupling. That difficulty is genuine. Ohsaka~\cite[Theorem~3.1]{Ohsaka} proved
that ``Determinant Maximization on arrowhead matrices is NP-hard and W[1]-hard
when parameterized by $k$,'' where the input is a rational positive-semidefinite
matrix and an integer $k$, the task is to find a $k\times k$ principal submatrix of
maximum determinant~\cite[Problem~2.1]{Ohsaka}, and an arrowhead matrix is one
whose support graph is a star. Because $\log$ is increasing, this problem is MESP,
and $k$ is our $s$. Although the theorem is stated for positive-semidefinite input,
the matrices that its reduction produces are positive definite. They are Gram
matrices of vectors $u_0,\dots,u_m\in\mathbb{Q}^{2m}$, of order $m+1$, such that $u_0$
is positive on coordinates $1,\dots,m$ and zero elsewhere, and, for $1\le i\le m$,
$u_i$ is positive on coordinates $i$ and $m+i$ and zero elsewhere: the vectors of
\cite[Eq.~(3.3)]{Ohsaka} have this sign pattern, because the parameters there and
the numbers $x_j$ are positive \cite[Problem~3.3]{Ohsaka}, and the rational vectors
of \cite[Lemma~3.7]{Ohsaka} approximate them entrywise within a factor $1\pm\varepsilon$
with $\varepsilon<1$, so they keep it. Coordinate $m+i$ is nonzero only in $u_i$\,, so
a linear combination of $u_0,\dots,u_m$ that vanishes has a zero coefficient on each
$u_i$ with $i\ge1$, and then on $u_0$\,, because $u_0\ne0$. Hence, these vectors are
linearly independent, their Gram matrix is positive definite, and MESP remains NP-hard, and W[1]-hard with respect to $s$, on
nonsingular $C$ for which $G(C)$ is a star; see also \cite{FampaLeeRecent2026}. We
draw two consequences.
\begin{enumerate}
\item[(i)] An $n$-vertex star is a spider with $r=n-1$ legs. Hence, unless
$\mathrm{P}=\mathrm{NP}$, if an algorithm for MESP on spiders runs in time at most
$n^{\omega(r)}$ on every spider with $r$ legs, for some function $\omega$, then $\omega$ is
unbounded; because $s\le n$, the same holds for bounds of the form $n^{\omega_0}s^{\omega(r)}$
with $\omega_0$ constant. For fixed $r$, Theorem~\ref{thm:spider} gives $O(ns^2+s^{r+1})$
time, whose exponent grows with $r$; some such growth is therefore unavoidable.
\item[(ii)] On a star, $|\mathcal{E}|=\sum_{k=0}^{s-1}\binom{n-1}{k}=n^{O(s)}$, so
Theorem~\ref{thm:spider} runs in time $n^{O(s)}$. By the W[1]-hardness result, we
cannot expect a running time of the form $f(s)\,n^{O(1)}$ on stars, unless the
parameterized complexity classes FPT and W[1] coincide.
\end{enumerate}
The greedy algorithm of \cite[Section~3]{ALTHANI2023120} for
arrowhead matrices applies precisely to the body-containing case above, and their
sufficient condition identifies instances on which the scalar coupling does not
disturb the greedy order.

\section{Closing remarks}\label{sec:closing}

\begin{remark}[All budgets at once]\label{rem:budgets}
The tables of Theorems~\ref{thm:main} and~\ref{thm:precision} contain the whole
vector $(z(C,t))_{t\le s}$\,. This vector is exactly the input that we need in order
to combine independent blocks $C^{(1)},\dots,C^{(m)}$ by a knapsack-type recursion,
\[
  \Gamma_j(t)=\max_u\{\Gamma_{j-1}(t-u)+z(C^{(j)},u)\},
\]
with $\Gamma_0(0):=0$ and $\Gamma_0(t):=-\infty$ for $t>0$, where $u$ ranges over
$0\le u\le\min(t,n_j)$ and $n_j$ is the order of $C^{(j)}$; then
$\Gamma_m(t)=z(C,t)$. This recursion covers so-called ``$k$-tridiagonal matrices'' (see \cite{daFonsecaYilmaz2015}),
whose support graph consists of $k$ disjoint paths. We use the same recursion for
spiders in Section~\ref{sec:spiders}.
\end{remark}

\begin{remark}[Beyond tridiagonal matrices]
Lemma~\ref{lem:pieces} is the reason that it suffices to consider $O(n)$ distinct
``last pieces'' per prefix. For $(2k+1)$-diagonal matrices with $k\ge2$, pieces
separated by fewer than $k$ unchosen indices interact, and the number of possible
pieces grows exponentially. Whether MESP is polynomially solvable for fixed
bandwidth $k\ge2$ appears to be wide open.
\end{remark}

\medskip
\noindent {\bf Acknowledgments.} 
M.~Fampa was supported in part by grant 304340/2026-0 from CNPq.
J.~Lee was supported in part by grant FA9550-22-1-0172 from AFOSR.

\medskip
\noindent {\bf CRediT authorship contribution statement.}
Both authors contributed equally in all aspects of this work.

\medskip
\noindent {\bf Data availability.}
No data was used for the research described in the article.

\medskip
\noindent {\bf Declaration of generative AI and AI-assisted technologies in the manuscript preparation process.}
During the preparation of this work, the authors used Claude to assist in developing and checking the work. 
The authors take full responsibility for the content of the published article.

\appendix

\section{Recognizing an unknown spider}\label{app:spider}
\setcounter{theorem}{0}
\renewcommand{\thetheorem}{\Alph{section}.\arabic{theorem}}

Here, $C$ is positive definite and given as a dense matrix, $Y$ is a Gaussian
random vector with covariance matrix $C$, and we write
$\rho_{ij}:=C_{ij}/\sqrt{C_{ii}C_{jj}}$ for $i\ne j$.

\begin{proposition}\label{prop:unknown}
In $O(n^2)$ time, we can decide whether $G(C^{-1})$ is a tree and, if it is, find
the tree, the diagonal of $C^{-1}$ and its entries on the edges of the tree, and
$\ldet C$. In particular, within the same time, we can decide whether $G(C^{-1})$ is
a spider.
\end{proposition}

\begin{proof}
We proceed in three steps.
\begin{enumerate}
\item[(1)] If $\rho_{ij}=0$ for some $i\ne j$, then we reject. Otherwise, by
Prim's algorithm for dense graphs, we compute in $O(n^2)$ time a spanning tree $T$
of the complete graph on $\Nn$ that has maximum weight for the weights
$|\rho_{ij}|$; the algorithm only compares these weights, each available in $O(1)$
time.
\item[(2)] We compute $Q_T$ and $\Lambda_T$ from the formulas that follow
Theorem~\ref{thm:spider}, with $T$ in place of the spider: $Q_T$ is supported on the
diagonal and the edges of $T$, and the computation costs $O(n)$ time.
\item[(3)] We check whether $Q_TC=I$. Row $i$ of $Q_T$ has $\deg_T(i)+1$ nonzero
entries, so the product costs $O\bigl(n\sum_i(\deg_T(i)+1)\bigr)=O(n^2)$ time. If the
check fails, then we reject; otherwise, we report $T$, $Q_T$\,, and
$\Lambda_T$\,. Deciding whether $T$ is a spider then costs $O(n)$ time.
\end{enumerate}

\emph{Soundness.} If the check in (3) passes, then $C^{-1}=Q_T$\,. The entry of
$Q_T$ on an edge $\{u,v\}$ of $T$ comes only from the $2\times2$ inverse for that
edge, so it equals $-C_{uv}/(C_{uu}C_{vv}-C_{uv}^2)$, which is nonzero by (1).
Hence, $G(C^{-1})=T$, and $\Lambda_T=\ldet C$ by the second formula.

\emph{Completeness.} Suppose that $G(C^{-1})$ is a tree $T^*$. As in the soundness
argument, the formulas for $T^*$ give
$(C^{-1})_{uv}=-C_{uv}/(C_{uu}C_{vv}-C_{uv}^2)$ for every edge $\{u,v\}$ of $T^*$, so
$\rho_{uv}\ne0$; also, $|\rho_{uv}|<1$ because $C$ is positive definite. Let
$i=k_0,k_1,\dots,k_m=j$ be the path in $T^*$ from $i$ to $j$. Because $k_1$ separates
$i$ from $j$ in $T^*$, the global Markov property of Gaussian graphical
models~\cite[Corollary~6.4, Proposition~6.6, and Theorem~4.6]{Lauritzen2026} gives
that $Y_i$ and $Y_j$ are conditionally independent
given $Y_{k_1}$\,, so $\rho_{ij}=\rho_{ik_1}\rho_{k_1j}$\,; by induction on $m$,
$\rho_{ij}=\prod_{l=1}^{m}\rho_{k_{l-1}k_l}$\,. Hence, no $\rho_{ij}$ is zero, and (1)
does not reject. Moreover, $d(i,j):=-\log\rho_{ij}^2$ is a tree metric on $T^*$ with
positive edge lengths: if $\{i,j\}$ is not an edge of $T^*$, then $m\ge2$, and
$d(i,j)=\sum_{l=1}^{m}d(k_{l-1},k_l)$ exceeds $d(k_{l-1},k_l)$ for every $l$. So every
edge outside $T^*$ is strictly the lightest, in the weights $|\rho_{ij}|$, on the
cycle that it closes with $T^*$, and $T^*$ is the unique maximum-weight spanning
tree. Hence, $T=T^*$ in (1), the formulas in (2) give $Q_T=C^{-1}$ and
$\Lambda_T=\ldet C$, and the check in (3) passes.
\end{proof}

The tree $T$ in (1) is the tree of~\cite{ChowLiu1968} (see
also~\cite[Corollary~8.4]{Lauritzen2026}), because the
mutual information $-\frac{1}{2}\log(1-\rho_{ij}^2)$ of $Y_i$ and $Y_j$ increases
with $|\rho_{ij}|$. When $G(C^{-1})$ is a path, this procedure gives another way to
find the order in Remark~\ref{rem:perm}. No algorithm can run in $o(n^2)$ time
in the worst case, because an entry of $C$ that it does not read could be changed,
by an amount small enough to keep $C$ positive definite, so that $G(C^{-1})$ is no
longer a tree.

\bibliographystyle{elsarticle-num}
\bibliography{MESP_update}

\end{document}